%% file: windfarm.tex
\documentclass[12pt]{amsart}
\pdfoutput=1 

\usepackage{graphicx}
\usepackage{amsmath}
\usepackage{amsfonts}
\usepackage{amsthm}
\usepackage{amssymb}
\usepackage{color} 
\usepackage[lined, vlined, norelsize]{algorithm2e} 

\usepackage{hyperref}

\newcommand{\spreadsheetlink}{
\\ \url{http://www.cs.umb.edu/~eb/windfarm/windfarm-v1-1.xlsm}}

\graphicspath{{}} 

\newcommand{\excel}[1]{{\texttt{#1}}}

\newtheorem{thm}{Theorem}
\newtheorem{theorem}[thm]{Theorem}
\newtheorem{corollary}[thm]{Corollary}

\title[Windfarm geometry]{Modeling how windfarm geometry
  affects bird mortality}

\author[E. Bolker]{Ethan\ D.\ Bolker}
\address{Departments of Mathematics and Computer Science, University
  of Massachusetts Boston, MA 02125 }
\email{eb@cs.umb.edu}

\author[J. Hatch]{Jeremy J. Hatch}
\address{Department of Biology, University of Massachusetts Boston, MA 02125}
\email{jeremy.hatch@umb.edu}

\author[C. Zara]{Catalin Zara}
\address{Department of Mathematics, University of Massachusetts
  Boston, MA 02125} 
\email{catalin.zara@umb.edu}

\date{August 6, 2014}

\begin{document}

\begin{abstract}
Birds flying across a region containing a windfarm risk death from
turbine encounters. This paper describes a geometric model that helps
estimate that risk and a spreadsheet that implements the model.
\end{abstract}

\maketitle

\section{Introduction}
After several years of controversy, Cape Wind 
will soon begin constructing a wind farm of 130 turbines, each about 100m in
diameter, spread over about 65 km$^2$ (25 square miles) in Nantucket
Sound off the coast of Massachusetts.

One component of the controversy is the potential for mortality of
birds that pass through the wind farm. This paper and the software it
describes is the result of a request from the biologist (Hatch) to the
mathematicians (Bolker and Zara) for help with some of the underlying
elementary geometry for modeling encounters with wind turbines during
such crossings. Given turbine locations, flight direction, and the
probability of a bird surviving a single passage through a turbine we
calculate the expected number of turbine encounters for each bird and
the probability of safe passage through the windfarm. To estimate
absolute mortality numbers you must combine these per bird estimates
with data about the number of birds exposed to the risk.

The most significant simplifying assumption is requiring a single
input parameter for the survival probability for a single encounter.
That number is hard to know. It depends on bird and turbine
characteristics, on bird behavior (\emph{e.g.} avoidance) and on
flight and wind speed. Band \cite{band2012} proposes a model that
predicts the probability of surviving an encounter based on 
these inputs.
Our model complements his: he pays careful attention to details at the
level of the individual birds and turbines, but does not deal with the
the arrangement of turbines in the farm. 

Our model is 
particularly straightforward mathematically. For some questions all
you need for good bounds or even exact answers (to the model) is a
calculator and the number and size of turbines. For others an Excel
spreadsheet (which we provide) does the job. 

One important insight to draw from our geometric model
is that birds passing through the wind farm turbine height, may
encounter surprisingly few turbines and that this number is 
probably greatly reduced by avoidance.
Chamberlain \emph{et.~al} 
\cite{chamberlain_effect_2006} show that Band model predictions
are much more sensitive to errors in estimating
avoidance behavior than to equivalent errors in all other
input parameters. Birds may act to
avoid both whole turbines individual blades, so may lower both the
average number of turbines encountered and the encounter mortality
probability. Since we assume no active avoidance, our mortality
estimates are likely to conservative -- that is, too high. 

Our model is generic -- it accepts turbine coordinates as input.
In this paper we apply it to a simple example that
makes the geometry and mathematics clear.  In
Section~\ref{sec:biology} we report on several studies that use it in
real situations. 

\section{The basic model}
We assume that each bird follows a \emph{path} $T$ when it flies at a
constant speed, height and heading (direction) across an area
containing a wind farm. We want to compute the mortality probability
$M(T)$ that the bird fails to survive its passage through the wind
farm. To that end, let $E(T)$ be the number of turbines the bird
encounters in its travel along $T$. Let $p$ be the probability of safe
passage through one turbine, and assume that surviving turbine
encounters are independent events. Then the probability of surviving
all the encounters is $p^{E(T)}$ so 
\begin{equation*}
    M(T) = 1 - p^{E(T)} .
\end{equation*}

The price for this simple computation is the unrealistic
assumptions we need to justify it.
The first is our requirement that paths be straight
lines. The second is our use of a single value $p$ for the survival
probability for any single encounter. In fact the value of $p$
depends not only
on the geometry of the encounter, thus on the size of the bird,
on the speed at which it crosses through a revolving rotor,
on the angle the flight path makes with the vertical plane
of the turbine and the distance from the center of the turbine at
which the encounter takes place,
but also on the active avoidance behavior of the bird.
So to use the model in any
particular case you must decide on an appropriate average value for
$p$ or, to compute an upper bound on mortality, a minimum value.

You can then use the model to
calculate the average values $\bar E$ of
$E(T)$ and $\bar M$ of $M(T)$ over an appropriate set of paths $\{ T
\}$.
Because finding and justifying a correct value for $p$ is
extremely difficult, our results are more reliable for $\bar E$ than
for $\bar M$.

Imagine a flock of birds crossing the wind farm on
the way from some distant point to some distant point. Each flies
along one path from the set $\{T\}$ of parallel paths on a particular
compass bearing $\theta$. 
We assume the paths cross
a line segment perpendicular to the line of flight uniformly
distributed along its length.
We write $\bar E = \bar E(\theta)$ and $\bar M = \bar
M(\theta)$ since both averages may depend on the angle $\theta$.


We compute averages both over paths that actually
cross the wind farm, and over paths that cross a small
circle that contains the wind farm,%
\footnote{
The circle is the smallest one containing the windfarm centered at the
average position of the turbines. It's not the smallest
circle containing the windfarm, which might have a different center,
but it is close to that circle.}
for heights uniformly distributed over a
specified vertical range.

Two figures illustrate the geometry of the small wind farm we will use
as an example.
The top view (Figure~\ref{fig:topView}) shows
four turbines with blade length $B$ meters. 
\begin{figure}[h]
\centering{
\resizebox{75mm}{!}{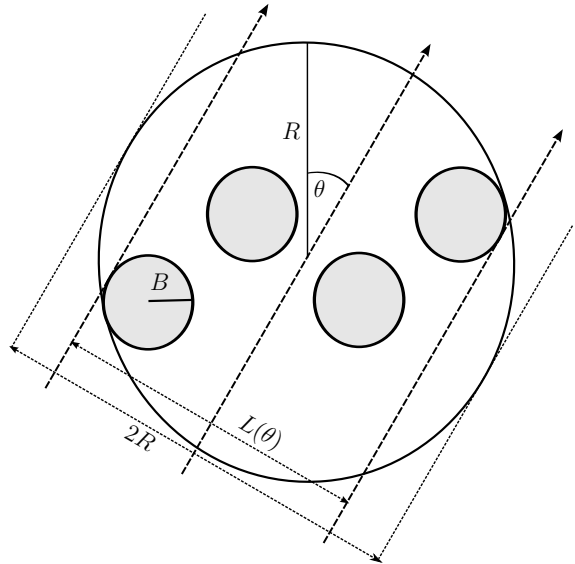}
\caption{Top view.}
\label{fig:topView}
}
\end{figure}
Turbine rotor planes can
themselves rotate about a vertical axis in order to present maximum
area to the prevailing wind; the small circles in the figure represent
the possible positions of the 
endpoints of the turbine rotors. This figure is not
drawn to scale - in a real wind farm the distance
between turbine centers would be on the order of $5$ to $15$ turbine
diameters rather than the approximately $1.5$ turbine diameters
shown. 
The dashed lines indicate the direction of flight, labeled
$\theta$, measured in degrees East of North. The small circle containing
the wind farm has radius $R$; the actual diameter of the wind farm when
birds fly on bearing $\theta$ is $L(\theta)$.

The side view (Figure~\ref{fig:sideView}) is what
a bird would see looking ahead about to enter the wind farm,
assuming for the moment that the flight is 
either upwind or downwind. Since the turbines rotate so that they
always face the wind each one appears to the bird as a circle.  White
regions correspond to paths that miss all the turbines, light gray
regions to paths that meet one and dark ones to paths that meet two.

\begin{figure}[ht]
\centering{
\resizebox{75mm}{!}{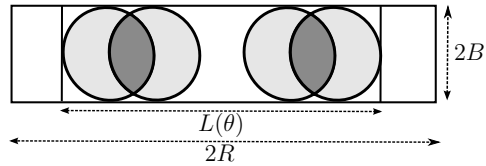}
\caption{Side view.}
\label{fig:sideView}
} 
\end{figure}

Then the average number of turbines encountered is just
\begin{equation*}
    \bar E(\theta) = \frac{2 \times  \text{dark area}
            + 1 \times \text{gray area}}
         {\text{total area}}
\end{equation*}
where the $ {\text{total area}}$ of the rectangle is either
$2BL(\theta)$ or $4BR$ depending on the set of paths you are
interested in. 

The average mortality probability is
\begin{equation}\label{eq:mortality}
    \bar M(\theta) = 1 - \frac{p^2 \times  \text{dark area}
            + p \times \text{gray area}
            + \text{white area}}
         {\text{total area}}
\end{equation}

It is clear that these areas and their generalizations for more than
two encounters can all be computed exactly using elementary
geometry and trigonometry. Our spreadsheet does that with an efficient
algorithm that runs in time proportional to the number of turbines.
That efficiency is possible because we assume $p$ is
independent of where a bird crosses a turbine. Chamberlain \emph{et.~al.}
\cite{chamberlain_appraisal_2005} discuss a similar model which requires 
numerical integration and the need to ``adjust for overlapping
rotors'' in order to deal with a more complex determination of $p$.

\section{Counting encounters, computing probabilities}

In this section we present some mathematics that shows that you can
compute $\bar E(\theta)$ and an upper bound for $\bar M(\theta)$
without needing even the elementary geometry required to calculate
the gray and dark areas in Figure~\ref{fig:sideView}.
You don't need our spreadsheet: a calculator will do. 

Imagine tacking targets (of any shape) to a dart board. Distribute
them as you wish. You may (in fact should) let them overlap. Then
suppose darts hit the dart board with a uniform distribution. Let $\bar
E$ be the average number of targets hit by a dart.

\begin{theorem}{\rm
\label{thm:avnumb}
\begin{equation*}
    \bar E = \frac{ \text{total area of targets}} {\text{area of
dart board}} .
\end{equation*}}
\end{theorem}

\begin{proof}
Write $X$ for the dart board.
For each target $t$ let
let $\chi_t$ be the characteristic function of $t$. That is,
$\chi_t(x)=1$ if $x$ is in $t$ and $0$ if it is not.
A dart landing at $x$ hits
$\sum_t \chi_t(x)$ targets, so the average number of
targets hit is
\begin{align*}
\bar E & =  \frac{ 1 } { {\text{area of }X} }\int_X \sum_t \chi_t(x)dx    
=  \frac{ 1 } { {\text{area of }X} }\sum_t \int_X \chi_t(x) dx \\
      & =  \frac{ 1 } { {\text{area of }X} }\sum_t \hbox{area of }t
=  \frac{ \text{total area of targets}} {\text{area of }X} . 
\end{align*}
\end{proof}

To model a wind farm, interpret paths as darts. We used the darts
metaphor in the theorem in order to capture its true geometric
generality. We find the theorem somehow simultaneously obvious and
counterintuitive, and so think it useful and informative to provide
this proof.%
\footnote{You can also view this theorem as a corollary of the fact
  that the expected value of a sum of random variables is the sum of
  their expectations.}

Consider a wind farm with $N$ turbines with blade length $B$.
Recall that $L(\theta)$ is the ``diameter'' of the wind farm region
perpendicular to bearing $\theta$ and $R$ is the radius of the small
circle containing the wind farm. 
Suppose bird flights are perpendicular to the
rotor planes and uniformly distributed vertically between the top and
bottom of the rotors.

\begin{corollary}
\label{cor:avnumb}{\rm
For flights that actually cross the wind farm on bearing $\theta$
\begin{equation}
\label{eq:avnumbtheta}
   \bar E(\theta) = \frac{N \pi B^2}{2BL(\theta)}  = \frac{N \pi B}{2L(\theta)}
\end{equation}
independent of turbine placement.
For flights that cross the small circle
\begin{equation}
\label{eq:avnumbR}
   \bar E = \bar E(\theta) = \frac{N \pi B^2}{(2B)(2R)} = \frac{N \pi B}{4R }
\end{equation}
independent both of $\theta$ and of turbine placement.
}
\end{corollary}

Imagine an onshore windfarm: $N$ turbines with centers $D$ meters
apart in a line along a ridge perpendicular to the prevailing
wind. For flights with or against the wind, 
each bird meets no turbines, or just one.
Since $2R \approx L(\theta)
\approx ND $,
\begin{equation*}
   \bar E(\theta) = \frac{N \pi B}{2L(\theta)} \approx
\frac{N \pi B}{2ND} = \frac{\pi}{2} \times \frac{B}{D}.
\end{equation*}
When $B/D$ is on the order of $1/15$ to $1/7$
(\cite{meyers2012optimal}),  
$\bar E$ is on the order of $0.1$ to $0.2$. Between 10 and 20 percent
of the birds encounter a turbine.

The average $\bar E$ is likely to be less than one 
for offshore windfarms as well, even though turbines tend to be
bunched and birds can meet more than one.
To see why, imagine $N$ turbines arranged on a square grid and roughly
filling a circle. Then the radius $R$ of the circle will be
approximately $\sqrt{N/\pi}D$, where $D$ is the distance between
turbine centers along grid lines. Then Equation~\ref{eq:avnumbR} implies
\begin{equation}\label{eq:ebarapprox}
    \bar E \approx \frac{ N \pi B}{ 4\sqrt{N/ \pi}D }
           = \frac{ \pi^{3/2}\sqrt{N}}{4} \times \frac{B}{D}.
\end{equation}
For $B/D = 1/10$, $\bar E < 1$ when $N < 52$.
On average, a bird encounters less than one turbine. 
For $B/D = 1/20$, $\bar E < 1$ when $N < 207$. 
You can confirm that using the \excel{Circle} worksheet in the
spreadsheet.
The Cape Wind installation will have 130 turbines in an area roughly a
rectangle twice as wide as high, with $B/D \approx
1/16$. If they were packed in a circle Equation~\ref{eq:ebarapprox}
would yield $\bar E \approx 0.85$. The spreadsheet calculations show
$\bar E \approx 0.6$ for the actual positions of the turbines.

Unfortunately, $\bar M$ is usually 
harder to come by. One case is easy.
\begin{theorem}\label{thm:easymortality}
When birds encounter at most one turbine on bearing $\theta$
\begin{equation}\label{eq:easymortality}
\bar M(\theta) = (1-p)\bar E(\theta).
\end{equation}
\end{theorem}
\begin{proof}
The probability that a bird encounters a turbine is $\bar
E(\theta)$. If it does, it dies with probability $(1-p)$.
\end{proof}

In general, 
\begin{equation}\label{eq:survivalprob}
\bar M = 1 - \frac{ 1 } { {\text{area of }X} }\int_X p^{\sum_t \chi_t(x)} dx .
\end{equation}
Fortunately, there's an
easy estimate for this integral that provides an upper bound for
mortality probability and a straightforward algorithm for computing
the integral exactly using simple geometry. We present the former
here. The latter is in the appendix and implemented in the
spreadsheet. 

\begin{theorem}
\label{th:mortprob}{\rm
\begin{equation*}
    \bar M \leq 1 - p ^ {\bar E }.
\end{equation*}}
\end{theorem}

\begin{proof}
Since the function $z \rightarrow p^z$ is convex,
Jensen's inequality implies that the average value of
$p^{\sum_t \chi_t(x)}$ is at least as large as $p^{\bar E}$.
\end{proof}

There's a second estimate that's also useful because in practice, the
survival probability $p$ is quite close to $1$.

\begin{corollary}
\label{cor:p1}{\rm
When $p \approx 1$,
\begin{equation*}
\bar M \approx (1-p)\bar E \; .
\end{equation*}}
\end{corollary}

\begin{proof}
Let $q = 1-p$.  Then $q \approx 0$ and for any real number $\alpha$
\begin{equation*}
    p^\alpha = (1-q)^\alpha = 1 -\alpha q
    +  \text{ lower order terms }
     \approx 1-\alpha q .
\end{equation*}
Setting $\alpha = \sum_t
\chi_t(x)$ and integrating to compute the 
average value of the left hand side finishes the proof.
\end{proof}

Theorem~\ref{thm:easymortality} says that this approximation is exact
when birds encounter at most one turbine. The ``lower order terms''
we've ignored deal with multiple encounters.
Figure~\ref{fig:estimateMbar} shows that it's a good
estimate for Cape Wind -- a real offshore windfarm -- even with an
unreasonably low survival probability of just 0.95. The (over)estimate
smooths out the small variations in the computed mortality
probabilities as a function of wind direction that are too precise to
have any useful meaning.
\begin{figure}[h]
\centering
\includegraphics[width=8cm]{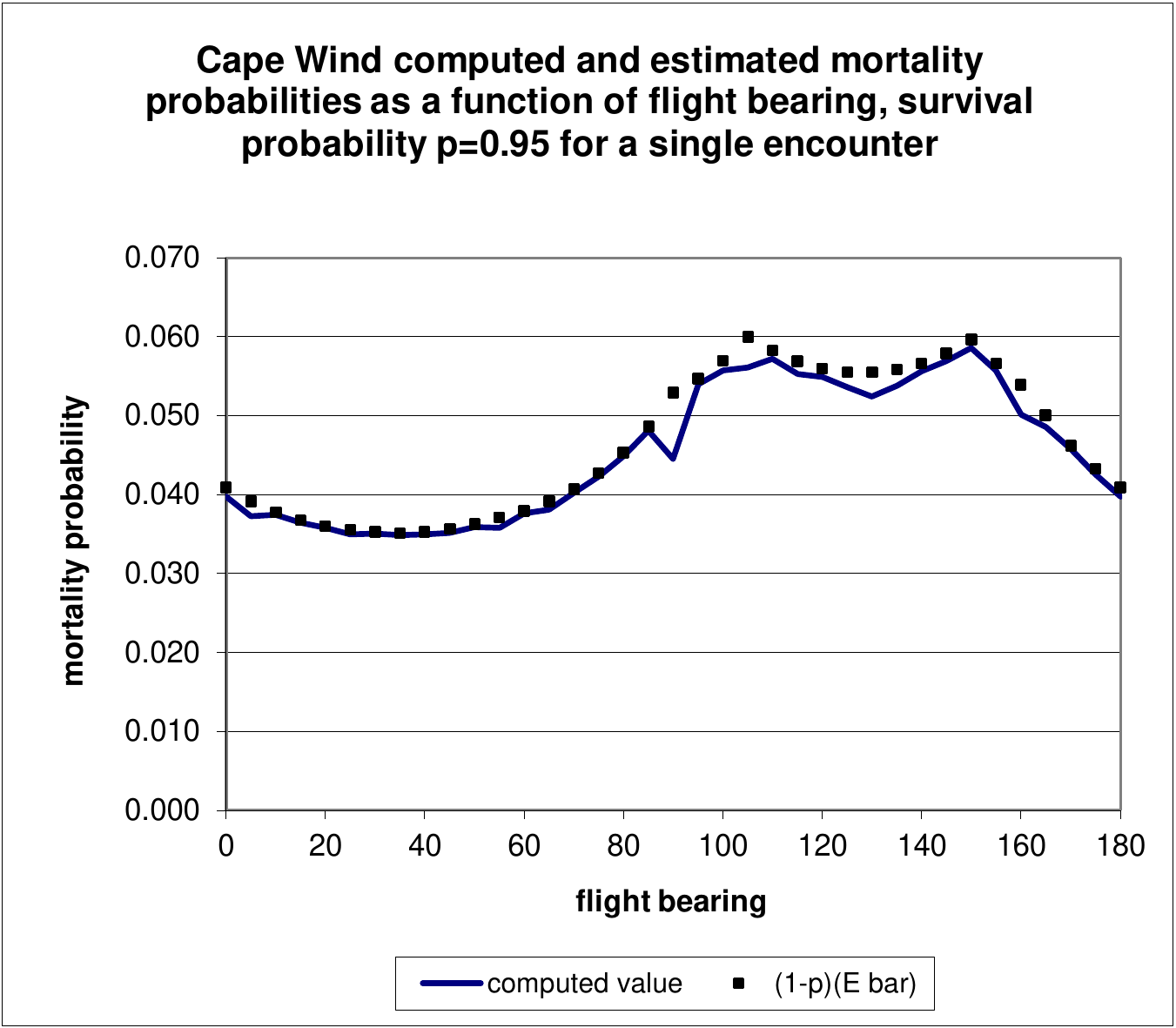}
\caption{Estimated and computed mortality probabilities.}
\label{fig:estimateMbar}
\end{figure}

\section{ Bells and whistles }

We can make the model more useful by adding a few more geometric input
parameters. First, you may specify a range of heights
at which birds are known to fly. If, for example, they tend to cross
the wind farm at an altitude between the center of the turbines and a
blade length above that center then the side view is shown in
Figure~\ref{fig:sideViewRange}. The white, gray and dark areas are
different but the computations for $\bar E(\theta)$ and $\bar
M(\theta)$ are the same.

\begin{figure}[ht]
\centering{
\resizebox{75mm}{!}{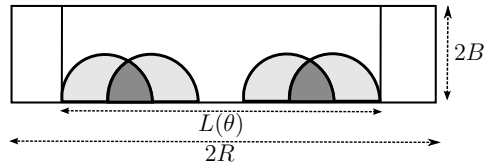}
\caption{Side view when birds tend to fly above turbine centers.}
\label{fig:sideViewRange}
} 
\end{figure}

A second feature allows you to evaluate the model when
birds fly at an angle to the wind. In that case
the birds see ellipses rather than circles. Their view is shown in
Figure~\ref{fig:sideViewEllipses}.

\begin{figure}[ht]
\centering{
\resizebox{75mm}{!}{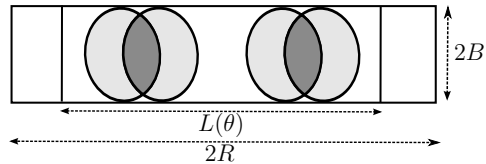}
\caption{Side view when birds fly at an angle to the wind.}
\label{fig:sideViewEllipses}
} 
\end{figure}

\section{The spreadsheet}

You will find the spreadsheet implementing our model at 
\spreadsheetlink{}. Here
we describe spreadsheet input and output and remind the user yet again
of some of our assumptions. 

\subsection*{Input}

Input goes in the yellow cells (with blue bold text) on the left in
the \excel{Main} worksheet. (Mouse over those cells to see documentation.)
Figure~\ref{fig:fourturbines} is a screen shot of that worksheet
showing values for the small four turbine wind farm we've been using
as an example. 
\begin{figure}[h]
\centering
\includegraphics[width=\textwidth]{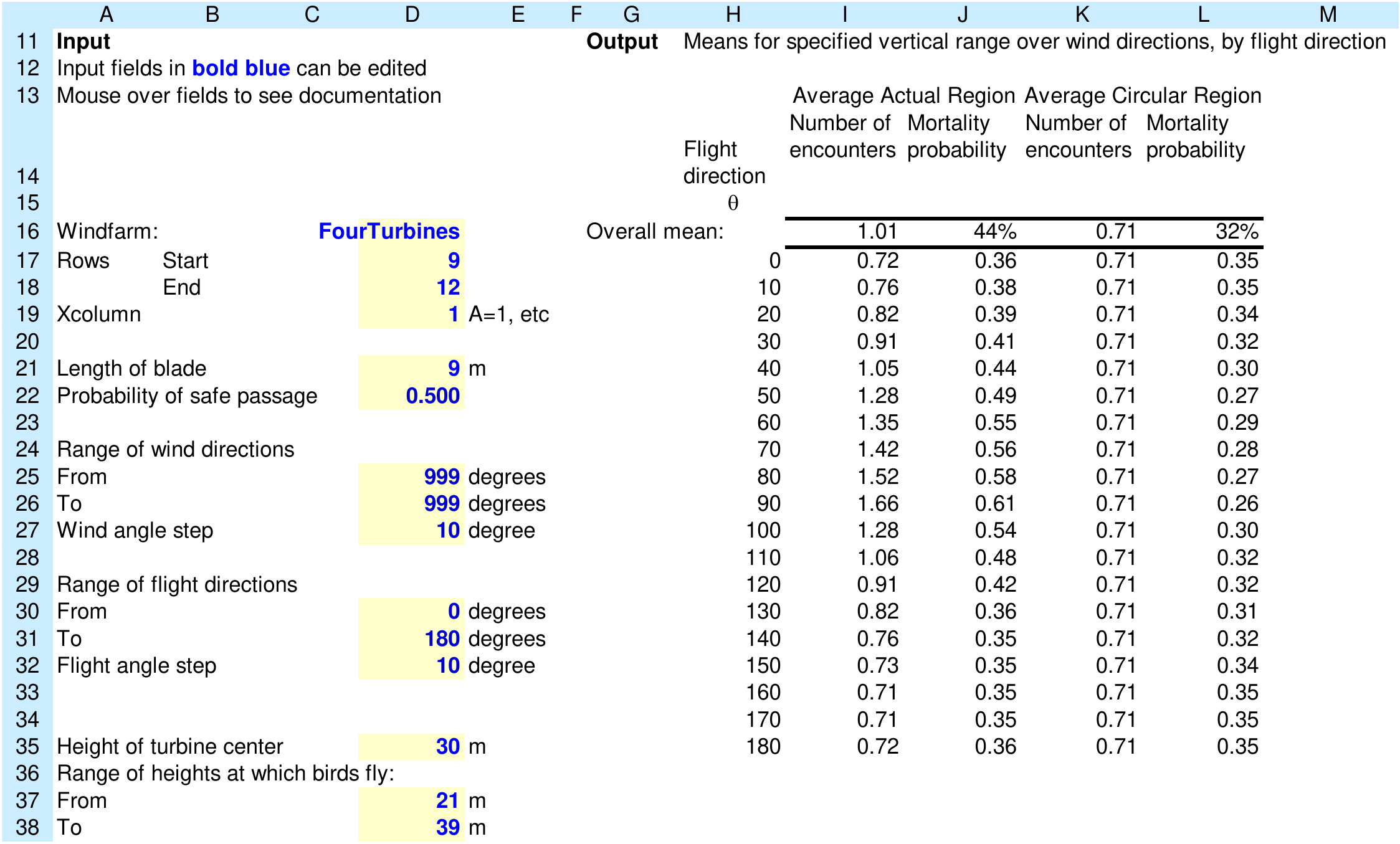}
\caption{Spreadsheet for four turbine model.}
\label{fig:fourturbines}
\end{figure}

\begin{itemize}
\item

Locations of turbines. Specify these in
a cartesian coordinate system with meters for units.
We entered them in the worksheet named \excel{FourTurbines} starting at cell
\excel{A9} and told the model that in cells \excel{D16:D19}.

To model your wind farm you may need to 
convert locations to a cartesian coordinate system. The
\excel{CapeWind} worksheet shows those calculations for the Cape Wind
farm. 

Figure~\ref{fig:fourturbinespositions} shows the locations of the four
turbines in our example. The input coordinates from the
\excel{FourTurbines} worksheet are on the right, the chart Excel drew
in in the \excel{Graphs} worksheet is on the left.

\begin{figure}[h]
\centering
\includegraphics[width=\textwidth]{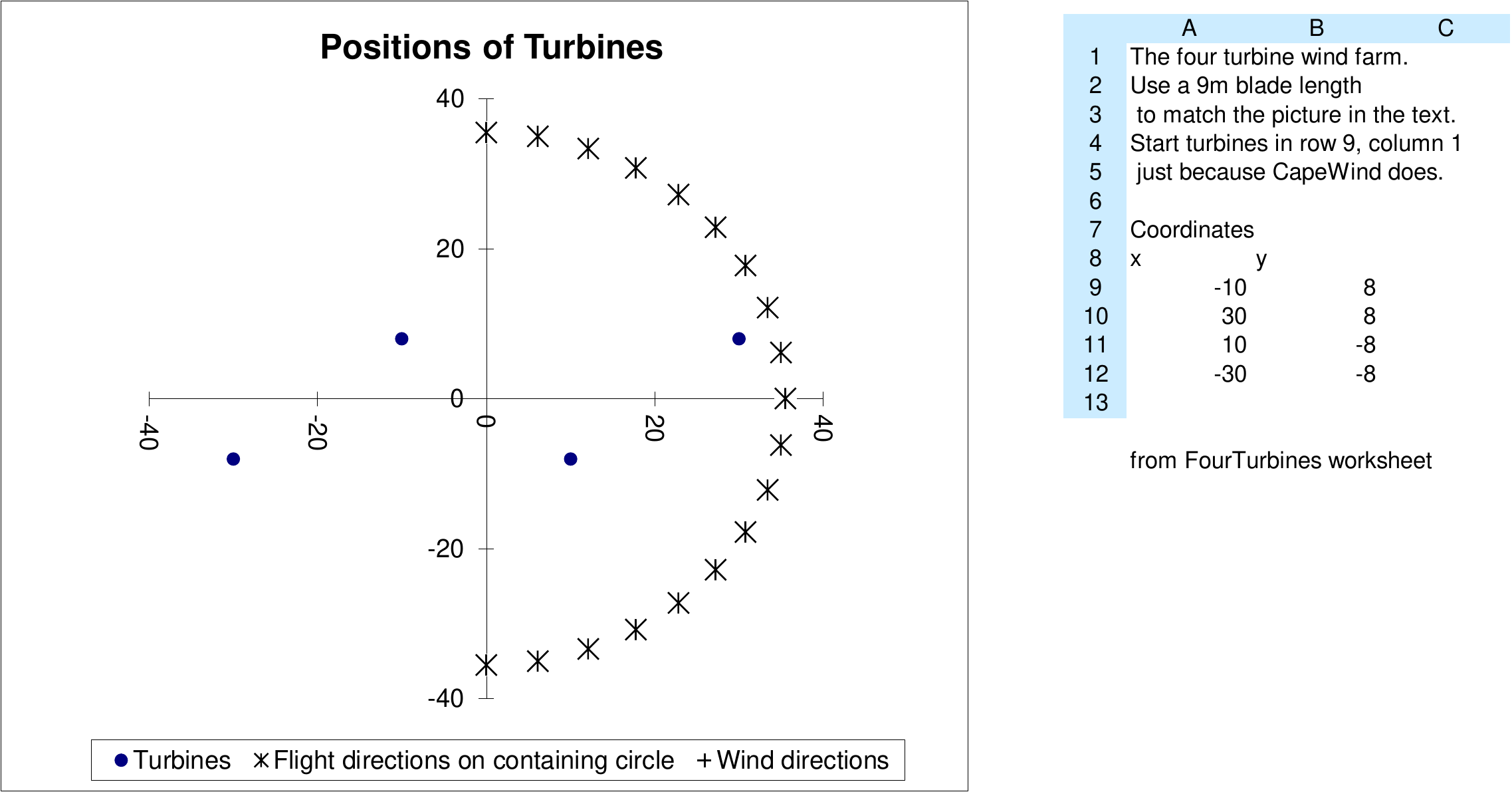}
\caption{Four turbine model coordinates.}
\label{fig:fourturbinespositions}
\end{figure}

\item
Turbine blade length $B$, in meters -- in cell \excel{D9}. That's 9m
in this example.

\item
The probability $p$ that a bird encountering a turbine survives the
encounter. We've entered $0.5$ in cell
\excel{D10}. That number is much too large for a real windfarm. We use
it here because it's easy to compute with so we will be able to check
the output of the spreadsheet by hand.

\item The height of the rotor centers, in meters -- in cell
  \excel{D35}. We use $30$m in our example.
\item
The upper and lower limits for the altitudes at which birds fly.
Setting these to \excel{-bladelength} and \excel{bladelength} respectively
leads to the side view in Figure~\ref{fig:sideViewRange}; we use
$30-9=21$ and $30+9=39$ meters.  Setting
both to 0 means that all birds fly at exactly the height of the
centers of the turbines.

\item
The compass bearings $\theta$ (in degrees east of north) for flight
paths you are interested in: minimum and maximum values and increment.
In the example we've entered $0$, $180$ and $10$ in cells
\excel{D30:D32}. 

\item
(Optional) The compass bearings $\omega$ (in degrees east of north)
for headings (directions) of the prevailing wind 
you are interested in:
minimum and maximum values and increment in cells \excel{D25:D27}.
Entering 999 as we have tells Excel to skip this computation and
assume that for the flight directions specified paths are
perpendicular to the plane in which the turbines rotate.

\end{itemize}

\subsection*{Output}

\begin{itemize}

\item
The radius of the smallest circle centered at $(0,0)$ surrounding the
wind farm, in cell \excel{D41}. For the \excel{FourTurbines} example
that's 41m. 

\item
For each flight heading $\theta$ (wind direction $\omega$)

\begin{itemize}

\item
the expected number $\bar E(\theta)$ 
of turbines encountered, and the average over $\theta$ ($\omega$).
\item
the expected probability
$\bar M(\theta)$ ($\bar M(\omega)$) of bird mortality
and the average over $\theta$ ($\omega)$.
\item
the maximum number
of encounters $E(T)$ for these paths and the corresponding
maximum mortality probability.
\end{itemize}

\item
Excel charts displaying this information.

Figure~\ref{fig:fourturbinescharts} shows the
results for flight directions in our example. You can see how they
reflect windfarm geometry: the peaks for the actual region near 90
degrees correspond to the fact that the four turbine centers are
approximately lined up West to East. The mortality values are large
because the probability of death for each encounter is an unreasonable
large 0.5. The number of encounters for the circular region is
constant (as predicted by Corollary~\ref{eq:avnumbR}). The mortality
chart dips near 90 degrees because most birds on that bearing crossing
the circle meet no turbines. Those few near the $x$-axis meet four;
they survive with probability $(1/2)^4 = 1/16$.

\begin{figure}[h]
\centering\
\includegraphics[width=\textwidth]{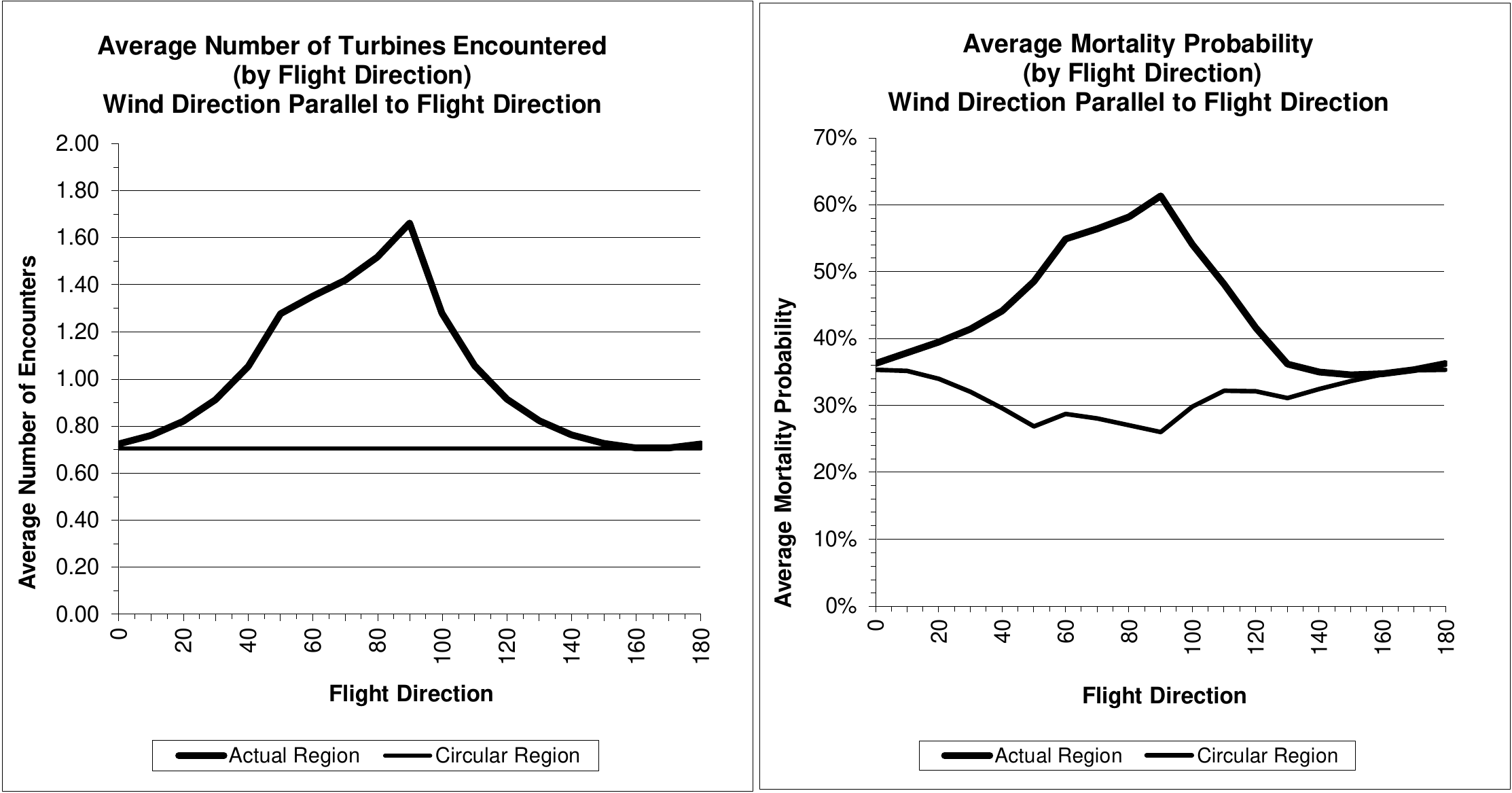}
\caption{Four turbine model output charts.}
\label{fig:fourturbinescharts}
\end{figure}

\end{itemize}

\subsection*{Tips and workarounds}

%
%
%
%
%
%
\begin{itemize}
\item
Before evaluating the model for your wind farm we
suggest you familiarize yourself with the way the model works
by playing with the four turbine example in the text and by building
your own small models in the \excel{test} worksheet in the
spreadsheet. Experiment with two or three turbines, and with survival
probabilities like $0$, $0.5$ and $1$ for which you will be able to
see that the answers are what you expect. 

\item The assumption that $p$ is constant and known in fact
  unreasonable. The model assumes that a correct average value has
  been computed for input, leaving to others the argument about how to
  compute that average. Evaluate your model several times using
  different estimates for $p$ to see how sensitive your results are to
  its value.

\item
When using the spreadsheet you should specify either a range of flight
headings and a single wind heading, or a single flight heading and
a range of wind headings. The spreadsheet will allow you to use two
ranges, but the results may be difficult to interpret.

\item
To model a nonuniform distribution of flight paths over
heights, perhaps with different mortality probabilities for each,
you can evaluate the model multiple times, once for each subrange over
which the height distribution is reasonably uniform, and combine the
results using Excel functions.

\item
To model mortality probabilities that depend on the angle
between the prevailing wind and the flight path you can evaluate the model
multiple times for restricted ranges of angles, varying the
probability as appropriate for each evaluation.

\item
To model mortality due to collisions with turbine
supports, set the blade length to the radius of the support, the
height range from 0 to 0 meters below and above the turbine center and
the wind bearing range from 999 to 999. 
The spreadsheet will then compute values for $\bar E$ and $\bar
M$ for paths at any height that might encounter the turbine supports.
The expected number $\bar E$ of encounters is likely to be small.
The value of $\bar M$ depends on the probability
$p$ of surviving an encounter, which may not be small.

\end{itemize}

\section{ Drawing Biological Conclusions }
\label{sec:biology}

Since we posted the first complete draft of this manuscript and
software in 2006 several papers have used it or referred to it. Here
we summarize some of that literature.

The results suggest (as we expected) that our model can be used to
provide a robust starting point for handling the geometry of the wind
farm but that most of the work required to estimate bird mortality is
in the biology - how many birds are there, where do they fly, and how
do they behave?

\subsection*{ The Nantucket Sound wind farm}

Jeremy Hatch and Solange Brault used our model in their analysis 
\cite{brault2007mortalities} of bird mortality for the proposed wind
farm on Nantucket Sound.

As we stressed in the introduction, modeling the probability of safe
passage through a wind farm requires two steps: estimating the number $E$
of turbines encountered , and the probability $p$ of surviving one
encounter. To estimate bird mortality in absolute terms requires a
further estimate of numbers of each bird-group of
interest at risk: those flying through the wind farm at
turbine-height.  Each component of these estimates has large
uncertainty and is likely to show great variation, much of which is
specific to the particular location and to the bird-group examined.

Hatch and Brault combine measured bird activity with robust estimation
methods for the difficult survival parameter $p$. They use Monte Carlo methods
to turn the single mortality probability estimates from our model into
mortality probability distributions, and run sensitivity analyses to
assess the importance of the estimates of each parameter.

Gordon \emph{et al} \cite{gordon}
use the Cape Wind configuration and measured
mortality from a nearby turbine at the Massachusetts Maritime Academy
to develop a more robust methodology for risk assessment:

\begin{quotation}
The Cape Wind modeling approach provides a foundation for exploring
the use of models in offshore conditions where high uncertainty
exists. The model developed by Bolker et al. (2006) is an example of a
model requiring minimal inputs, employing simple geometry and basic
probability theory to estimate avian mortality. This paper expands
upon the original work of Bolker by directly incorporating
observations of turbine avoidance behavior by terns into the published
mathematical framework. In addition, we modify the Bolker framework by
formally incorporating a risk based approach to decision making based
on the model outputs, including the use of a formal uncertainty
analysis. 
\end{quotation}

\subsection*{The Belgian Part of the North Sea}

Nicolas Vanermen and Eric W.M. Stienen \cite{vanermen_2009} studied
bird mortality for a  proposed wind farm in the Belgian Part of the
North Sea. They used our model to find a worst case estimate of the
number of turbines encountered.

%
%

\appendix
\section{Algorithms}

Here we provide an algorithm to
evaluate the integral in Equation~\ref{eq:survivalprob}, which is the
formal statement of the numerator in 
Equation~\ref{eq:mortality}.
It's an analogue of Theorem~\ref{thm:avnumb}, true for some arrangements of
targets on a dart board -- fortunately, the ones we are interested in.

Let $D_1, D_2, \ldots, D_N$ be a sequence of disks in the plane with
the same radius and collinear centers, arranged in numerical order
along the line of centers. Then for each $i$, 
\begin{equation}\label{eq:intersections}
D_i \supseteq D_i \cap D_{i+1}
\supseteq D_i \cap D_{i+2}  \ldots .
\end{equation}
This is just what you can see in
Figure~\ref{fig:sideView}. It's also true for the shaded regions in
Figure~\ref{fig:sideViewRange}. It's exactly what we need for the next
theorem.

\begin{theorem}
Let $D_1, D_2, \ldots, D_N$ be a sequence of plane regions for which
Equation~\ref{eq:intersections} is true.
Let $\chi_i$ be the characteristic function of $D_i$,
$X_i = D_i \cup D_{i+1} \cup \cdots \cup D_N$ and $X=X_1$
Then
\begin{equation}\label{eq:probaccum}
  \int_X p^{\sum_{i=1}^N \chi_i(x)} dx =
\sum_{i=1}^N \sum_{j=i}^{N} \lambda_{j-i}(p) A_{ij}\; ,
\end{equation}
where $A_{ij} = \text{Area} (D_i \cap D_j)$ and
\begin{equation}
\lambda_{r} = \lambda_{r}(p) = \begin{cases}
p, & \text{ if } r=0,\\
p^2-2p, & \text{ if } r=1, \\
p^{r-1}(p-1)^2, & \text{ if } r \geqslant 2\, .
\end{cases}
\end{equation}

\end{theorem}

\begin{proof}
\begin{align*}
  \int_X p^{\sum_{i=1}^N \chi_i(x)} &\;  dx -
  \int_{X_2} p^{\sum_{i=2}^N \chi_i(x)}\; dx =  p(A_{11}-A_{12}) \\ 
& + (p^2-p)(A_{12}-A_{13})+ (p^3-p^2)(A_{13}-A_{14}) \\
& + \dotsb + (p^{k-1}-p^{k-2})(A_{1,k-1}-A_{1,k})+ (p^N-p^{N-1})A_{1,N}\\
& = pA_{11} + (p^2-2p)A_{12}+ (p^3-2p^2+p)A_{13} + \dotsb \\
& = \lambda_{0}A_{11} + \lambda_{1} A_{12} 
+ \lambda_{2}A_{13}+ \dotsb + \lambda_{k-1}A_{1N}  = \sum_{j=1}^{N} \lambda_{j-1}A_{1j}\; .
\end{align*}
Similarly
\begin{align*}
  \int_{X_2} p^{\sum_{i=2}^N \chi_i(x)} dx -
  \int_{X_3} p^{\sum_{i=3}^N \chi_i(x)} dx & = \sum_{j=2}^{N}   \lambda_{j-2}A_{2j} \\
  \int_{X_3} p^{\sum_{i=3}^N \chi_i(x)} dx  -
  \int_{X_4} p^{\sum_{i=4}^N \chi_i(x)} dx & = \sum_{j=3}^N \lambda_{j-3}A_{3j} \\
\vdots \qquad  & \quad \vdots \\
  \int_{X_N} p^{\sum_{i=N}^N \chi_i(x)} dx & = \sum_{j=N}^N \lambda_{j-N}A_{Nj} \\
\end{align*}
Adding these telescoping equations leads to \eqref{eq:probaccum}.
\end{proof}

The last piece is the computation of the area of the intersections of
the regions $D_i \cap D_j$. The $D_i$ are circles when the flight
direction is up or downwind. When the flight is at an angle to the wind,
the $D_i$ are ellipses obtained by compressing circles along a
diameter, with the same compression factor for each. Then we can find
the areas of the intersections by decompressing to circles, finding
those areas and then compressing the results.
 
Suppose $C_1$ and $C_2$ are circles of radius $R$, with centers
$\text{O}_1$ and $\text{O}_2$ situated at a distance
$\text{O}_1\text{O}_2 = 2d \leqslant 2R$. Let $D_1$ 
and $D_2$ be the segments of the disks bounded by lines parallel to
the line of the centers, at distance $a$ and $b$ from the line of
centers; these distances are positive if the line is above the line of
centers and negative otherwise. The area of the intersection $D_1 \cap
D_2$ can be computed as follows using the elementary geometry
illustrated in Figure~\ref{fig:intersectionArea}.

\begin{figure}[h]
\centering{
\resizebox{75mm}{!}{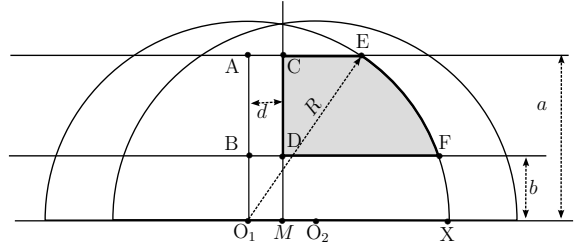}
\caption{Computing area between flight altitudes}
\label{fig:intersectionArea}
}
\end{figure}
\begin{equation*}
\text{Area} (D_1 \cap D_2) =  2\text{Area(CDFE)}
= 2 (\text{Area(CMXE)} - \text{Area(DMXF)})
\end{equation*}
and 
\begin{align*}
\text{Area(CMXE)} = & \text{Area(EO}_1\text{X)}
                    +\text{Area(AO}_1\text{E)}
                    -\text{Area(AO}_1\text{MC)} \\
= & \frac{1}{2} R^2 \arcsin{\left( \frac{a}{R}\right)} 
           + \frac{1}{2}a\sqrt{R^2-a^2} - ad \\
\text{Area(DMXF)} 
= & \frac{1}{2} R^2 \arcsin{\left( \frac{b}{R}\right)} 
           + \frac{1}{2}b\sqrt{R^2-b^2} -bd
\; .
\end{align*}
Therefore
\begin{align*}
\text{Area} (D_1 \cap D_2) = 
&  R^2 \left( \arcsin(a/R) - \arcsin(b/R) \right)  \\
& + a\sqrt{R^2-a^2} - b\sqrt{R^2 - b^2}
- 2d(a-b)\; . 
\end{align*}

The formula is valid when
\begin{equation*}
-\sqrt{R^2 - d^2}
 \leqslant b \leqslant a \leqslant 
\sqrt{R^2 - d^2} \; ;
\end{equation*}
if $a$ or $b$ lies outside that interval, replace it by the nearest
endpoint of the interval.

Here is the pseudocode for the Excel macro implementing that
algorithm in the spreadsheet. \verb!Area! computes the area of 
each intersection of elliptic sections. The inner loop terminates
prematurely as soon as the intersection is empty. 

\begin{algorithm}[H]
\DontPrintSemicolon
\SetAlgoLined
\SetKwRepeat{Repeat}{repeat}{until}%
\SetKwFunction{Area}{Area}
\SetKwFunction{lambda}{lambda}

/* Read input parameters */\;
\begin{itemize}
\item Location of turbines\;
\item Geometry of turbines (rotor height, blade length)\;
\item Range of flight directions (start, end, step)\;
\item Range of wind directions (start, end, step)\;
\item Probability of safe passage through one turbine\;
\end{itemize}

/* Normalize coordinates of turbines*/\; 
Compute centroid of wind turbines\;
Compute radius of enclosing region\;

\For{each flight direction}{
Project turbines along the flight direction\;
Sort coordinates of projections\;
\For{each wind direction}{
$ProbAccum \longleftarrow 0$\;
$AreaAccum \longleftarrow 0$\;
\For{$i \longleftarrow 1$ \KwTo $NumTurb$}
{
$AreaAccum \longleftarrow AreaAccum + \Area(D_i)$\;
$j \longleftarrow i$\;
\Repeat{$WedgeArea=0$ or $j = NumTurb +1$}
{
$WedgeArea \longleftarrow \Area(D_i \cap D_j)$\;
$Weight \longleftarrow \lambda(j-i, ProbSafe)$\;
$ProbAccum \longleftarrow ProbAccum + Weight * WedgeArea$\;
$j \longleftarrow j+1$\;
\If {j=i+1}
	{
	$AreaAccum \longleftarrow AreaAccum -\Area(D_i \cap D_j)$
	}
}
}
Record ProbAccum\;
Record AreaAccum
}
}
\end{algorithm}

\nocite{*}

\bibliographystyle{abbrv}      
\bibliography{windfarm_references}  

\end{document}

%% file: topView.pdf_tex
\begingroup%
  \makeatletter%
  \providecommand\color[2][]{%
    \errmessage{(Inkscape) Color is used for the text in Inkscape, but the package 'color.sty' is not loaded}%
    \renewcommand\color[2][]{}%
  }%
  \providecommand\transparent[1]{%
    \errmessage{(Inkscape) Transparency is used (non-zero) for the text in Inkscape, but the package 'transparent.sty' is not loaded}%
    \renewcommand\transparent[1]{}%
  }%
  \providecommand\rotatebox[2]{#2}%
  \ifx\svgwidth\undefined%
    \setlength{\unitlength}{279.4389283bp}%
    \ifx\svgscale\undefined%
      \relax%
    \else%
      \setlength{\unitlength}{\unitlength * \real{\svgscale}}%
    \fi%
  \else%
    \setlength{\unitlength}{\svgwidth}%
  \fi%
  \global\let\svgwidth\undefined%
  \global\let\svgscale\undefined%
  \makeatother%
  \begin{picture}(1,0.99211271)%
    \put(0,0){\includegraphics[width=\unitlength]{topView.pdf}}%
    \put(0.54025494,0.6490391){\color[rgb]{0,0,0}\makebox(0,0)[lb]{\smash{$\theta$}}}%
    \put(0.48440922,0.75166662){\color[rgb]{0,0,0}\makebox(0,0)[lb]{\smash{$R$}}}%
    \put(0.40538874,0.24147164){\color[rgb]{0,0,0}\rotatebox{-30.0000006}{\makebox(0,0)[lb]{\smash{$L(\theta)$}}}}%
    \put(0.25007316,0.48308714){\color[rgb]{0,0,0}\makebox(0,0)[lb]{\smash{$B$}}}%
    \put(0.20244571,0.22343363){\color[rgb]{0,0,0}\rotatebox{-30.0000006}{\makebox(0,0)[lb]{\smash{$2R$}}}}%
  \end{picture}%
\endgroup%

%% file: sideView.pdf_tex
\begingroup%
  \makeatletter%
  \providecommand\color[2][]{%
    \errmessage{(Inkscape) Color is used for the text in Inkscape, but the package 'color.sty' is not loaded}%
    \renewcommand\color[2][]{}%
  }%
  \providecommand\transparent[1]{%
    \errmessage{(Inkscape) Transparency is used (non-zero) for the text in Inkscape, but the package 'transparent.sty' is not loaded}%
    \renewcommand\transparent[1]{}%
  }%
  \providecommand\rotatebox[2]{#2}%
  \ifx\svgwidth\undefined%
    \setlength{\unitlength}{279.61171875bp}%
    \ifx\svgscale\undefined%
      \relax%
    \else%
      \setlength{\unitlength}{\unitlength * \real{\svgscale}}%
    \fi%
  \else%
    \setlength{\unitlength}{\svgwidth}%
  \fi%
  \global\let\svgwidth\undefined%
  \global\let\svgscale\undefined%
  \makeatother%
  \begin{picture}(1,0.28138407)%
    \put(0,0){\includegraphics[width=\unitlength]{sideView.pdf}}%
    \put(0.47305874,0.00294773){\color[rgb]{0,0,0}\makebox(0,0)[lb]{\smash{$2R$}}}%
    \put(0.914826,0.18196214){\color[rgb]{0,0,0}\makebox(0,0)[lb]{\smash{$2B$}}}%
    \put(0.4615372,0.05688205){\color[rgb]{0,0,0}\makebox(0,0)[lb]{\smash{$L(\theta)$}}}%
  \end{picture}%
\endgroup%

%% file: sideViewRangeShaded.pdf_tex
\begingroup%
  \makeatletter%
  \providecommand\color[2][]{%
    \errmessage{(Inkscape) Color is used for the text in Inkscape, but the package 'color.sty' is not loaded}%
    \renewcommand\color[2][]{}%
  }%
  \providecommand\transparent[1]{%
    \errmessage{(Inkscape) Transparency is used (non-zero) for the text in Inkscape, but the package 'transparent.sty' is not loaded}%
    \renewcommand\transparent[1]{}%
  }%
  \providecommand\rotatebox[2]{#2}%
  \ifx\svgwidth\undefined%
    \setlength{\unitlength}{279.61171875bp}%
    \ifx\svgscale\undefined%
      \relax%
    \else%
      \setlength{\unitlength}{\unitlength * \real{\svgscale}}%
    \fi%
  \else%
    \setlength{\unitlength}{\svgwidth}%
  \fi%
  \global\let\svgwidth\undefined%
  \global\let\svgscale\undefined%
  \makeatother%
  \begin{picture}(1,0.28013234)%
    \put(0,0){\includegraphics[width=\unitlength]{sideViewRangeShaded.pdf}}%
    \put(0.47305874,0.00294773){\color[rgb]{0,0,0}\makebox(0,0)[lb]{\smash{$2R$}}}%
    \put(0.914826,0.18196214){\color[rgb]{0,0,0}\makebox(0,0)[lb]{\smash{$2B$}}}%
    \put(0.4615372,0.05688205){\color[rgb]{0,0,0}\makebox(0,0)[lb]{\smash{$L(\theta)$}}}%
  \end{picture}%
\endgroup%

%% file: sideViewEllipsesShaded.pdf_tex
\begingroup%
  \makeatletter%
  \providecommand\color[2][]{%
    \errmessage{(Inkscape) Color is used for the text in Inkscape, but the package 'color.sty' is not loaded}%
    \renewcommand\color[2][]{}%
  }%
  \providecommand\transparent[1]{%
    \errmessage{(Inkscape) Transparency is used (non-zero) for the text in Inkscape, but the package 'transparent.sty' is not loaded}%
    \renewcommand\transparent[1]{}%
  }%
  \providecommand\rotatebox[2]{#2}%
  \ifx\svgwidth\undefined%
    \setlength{\unitlength}{279.61171875bp}%
    \ifx\svgscale\undefined%
      \relax%
    \else%
      \setlength{\unitlength}{\unitlength * \real{\svgscale}}%
    \fi%
  \else%
    \setlength{\unitlength}{\svgwidth}%
  \fi%
  \global\let\svgwidth\undefined%
  \global\let\svgscale\undefined%
  \makeatother%
  \begin{picture}(1,0.28138407)%
    \put(0,0){\includegraphics[width=\unitlength]{sideViewEllipsesShaded.pdf}}%
    \put(0.47305874,0.00294773){\color[rgb]{0,0,0}\makebox(0,0)[lb]{\smash{$2R$}}}%
    \put(0.914826,0.18196214){\color[rgb]{0,0,0}\makebox(0,0)[lb]{\smash{$2B$}}}%
    \put(0.4615372,0.05688205){\color[rgb]{0,0,0}\makebox(0,0)[lb]{\smash{$L(\theta)$}}}%
  \end{picture}%
\endgroup%

%% file: wedge-area.pdf_tex
\begingroup%
  \makeatletter%
  \providecommand\color[2][]{%
    \errmessage{(Inkscape) Color is used for the text in Inkscape, but the package 'color.sty' is not loaded}%
    \renewcommand\color[2][]{}%
  }%
  \providecommand\transparent[1]{%
    \errmessage{(Inkscape) Transparency is used (non-zero) for the text in Inkscape, but the package 'transparent.sty' is not loaded}%
    \renewcommand\transparent[1]{}%
  }%
  \providecommand\rotatebox[2]{#2}%
  \ifx\svgwidth\undefined%
    \setlength{\unitlength}{330.65651855bp}%
    \ifx\svgscale\undefined%
      \relax%
    \else%
      \setlength{\unitlength}{\unitlength * \real{\svgscale}}%
    \fi%
  \else%
    \setlength{\unitlength}{\svgwidth}%
  \fi%
  \global\let\svgwidth\undefined%
  \global\let\svgscale\undefined%
  \makeatother%
  \begin{picture}(1,0.43054692)%
    \put(0,0){\includegraphics[width=\unitlength]{wedge-area.pdf}}%
    \put(0.53909058,0.22446307){\color[rgb]{0,0,0}\rotatebox{53.14725992}{\makebox(0,0)[lb]{\smash{$R$}}}}%
    \put(0.3972524,0.01022095){\color[rgb]{0,0,0}\rotatebox{-0.07790739}{\makebox(0,0)[lb]{\smash{$\text{O}_1$}}}}%
    \put(0.52681164,0.00941239){\color[rgb]{0,0,0}\rotatebox{-0.07790739}{\makebox(0,0)[lb]{\smash{$\text{O}_2$}}}}%
    \put(0.43935472,0.22379754){\color[rgb]{0,0,0}\rotatebox{-0.07790739}{\makebox(0,0)[lb]{\smash{$d$}}}}%
    \put(0.46812064,0.00929963){\color[rgb]{0,0,0}\rotatebox{-0.07790739}{\makebox(0,0)[lb]{\smash{$M$}}}}%
    \put(0.38204282,0.16705642){\color[rgb]{0,0,0}\rotatebox{-0.07790739}{\makebox(0,0)[lb]{\smash{$\text{B}$}}}}%
    \put(0.382429,0.30364299){\color[rgb]{0,0,0}\rotatebox{-0.07790739}{\makebox(0,0)[lb]{\smash{$\text{A}$}}}}%
    \put(0.49272723,0.30278308){\color[rgb]{0,0,0}\rotatebox{-0.07790739}{\makebox(0,0)[lb]{\smash{$\text{C}$}}}}%
    \put(0.61493081,0.34806911){\color[rgb]{0,0,0}\rotatebox{-0.07790739}{\makebox(0,0)[lb]{\smash{$\text{E}$}}}}%
    \put(0.49263825,0.17180659){\color[rgb]{0,0,0}\rotatebox{-0.07790739}{\makebox(0,0)[lb]{\smash{$\text{D}$}}}}%
    \put(0.76125616,0.16787329){\color[rgb]{0,0,0}\rotatebox{-0.07790739}{\makebox(0,0)[lb]{\smash{$\text{F}$}}}}%
    \put(0.935935,0.23824851){\color[rgb]{0,0,0}\rotatebox{-0.07790739}{\makebox(0,0)[lb]{\smash{$a$}}}}%
    \put(0.92233952,0.08865977){\color[rgb]{0,0,0}\rotatebox{-0.07790739}{\makebox(0,0)[lb]{\smash{$b$}}}}%
    \put(0.76381939,0.0052151){\color[rgb]{0,0,0}\rotatebox{-0.07790739}{\makebox(0,0)[lb]{\smash{$\text{X}$}}}}%
  \end{picture}%
\endgroup%